\documentclass{article}

\usepackage{amsmath, amsthm, amssymb}

\newtheorem{thm}{Theorem}
\newtheorem{cor}[thm]{Corollary}

\newtheorem{lem}[thm]{Lemma}

\theoremstyle{definition}
\newtheorem{defi}[thm]{Definition}

\newcommand{\rs}{\mathrm{row space}}
\newcommand{\rk}{\mathrm{rk}}
\newcommand{\lift}{\mathrm{lift}}
\newcommand{\F}{\mathbb{F}}

\newcommand{\N}{\mathbb{N}}
\newcommand{\G}{\mathcal{G}_q(k,n)}

\newcommand{\Vvs}{\mathcal{V}}
\newcommand{\Uvs}{\mathcal{U}}

\newcommand{\C}{\mathcal{C}}

\title{A lower bound for constant dimension codes from multi-component lifted MRD codes}

\author{Anna-Lena Trautmann\\ Institute of Mathematics\\ University of Zurich, Switzerland}

\date{}

\begin{document}
 
\maketitle


\section{Introduction}

In this work we investigate unions of lifted MRD codes of a fixed dimension and minimum distance and derive an explicit formula for the cardinality of such codes. This will then imply a lower bound on the cardinality of constant dimension codes.

We will first repeat some known results needed in this section. In Section 2 we will explain the construction and derive the formula for the cardinality and the lower bound. We conclude in Section 3.

Let $\F_q$ be the finite field with $q$ elements. $\G$ denotes the set of all $k$-dimensional vector spaces over $\F_q^n$, called the Grassmannian. A constant dimension code is simply a subset of $\G$. A metric on $\G$ is the injection metric, given by
\[d_I(\Uvs, \Vvs)= k - \dim(\Uvs \cap \Vvs)\]
for any $\Uvs, \Vvs \in \G$. The minimum injection distance of a code $\C \subseteq \G$ is defined as
\[d_I(\C) = \min \{d_I(\Uvs,\Vvs)\mid \Uvs, \Vvs \in \C , \Uvs \neq \Vvs\}.\]
By $A_q(n,d,k)$ we denote the maximal size of a code $\C \subseteq \G$ with minimum injection distance $d$.

There is a complete theory of matrix codes with the rank distance, which can be used to construct constant dimension codes. We will now give a brief overview on the most important definitions and results of this topic.

Let $A,B \in \F_q^{m\times n}$ be two matrices of the same size. It holds that 
$$d_R:=\rk (A -B)$$ 
is a metric on  $\F_q^{m\times n}$, called the rank metric. A rank-metric code is simply a subset of $\F_q^{m\times n}$. The minimum distance is defined in the usual way.

The following two theorems can be found in \cite{ga85a}:

\begin{lem}\label{thm:mrd}
Let $m, n \geq d$ and $C\in \F_q^{m\times n}$ be a rank-metric code with minimum rank distance $d$. Then 
\[|C| \leq q^{\max(n,m)(\min(n,m)-d+1)} .\]
\end{lem}

\begin{lem}\label{thm1}
For any set of parameters $n, m\geq d \in \N$ and arbitrary field size there exist codes attaining the bound of Lemma \ref{thm:mrd}. These codes are called \emph{maximum rank distance (MRD) codes}. 
\end{lem}

We will now explain how to use MRD codes for the construction of constant dimension codes.

\begin{defi}
For a given rank-metric code $C\in \F_q^{k\times n}$ the set
\[\lift(C) := \{\rs \left[\begin{array}{cc} I_{k} & A \end{array}\right] | A\in C\}\]
is called the \emph{lifting} of $C$. 
\end{defi}

\begin{lem}\cite{si08j}\label{lem:lift1}
If $C\in \F_q^{k\times n-k}$ is an MRD code of minimum rank distance $d$, then $\lift(C) \in \G$ is a constant dimension code with minimum injection distance $d$. Moreover,
\[|\lift(C)| = q^{(n-k)(k-d +1)} .\]
\end{lem}

Note, that this construction corresponds to the Reed-Solomon-like construction from K\"otter and Kschischang \cite{ko08}. We do not want to explain this construction in detail here, but the interested reader is referred to \cite{ko08}.

\section{Multi-component lifted MRD codes}

Naturally, appending $0$-columns in front of all code elements does not change the minimum distance. Thus, if $C\in \F_q^{k\times n-k-l}$ is an MRD code of minimum rank distance $d$, then
\[\rs\left[\begin{array}{lll}0_{k\times l} & I_{k\times k} & M \end{array}\right] | M \in C \} \subseteq \G\]
is a constant dimension code with minimum injection distance $d$ and cardinality  $q^{(n-k-l)(k-d+1)}$. This fact can be used to construct even larger codes, which has also been observed by, among others, \cite{ga11,sk08u}. We will now give our own formulation of a multi-component construction and derive an exact formula for the cardinality of these codes, which we call the \emph{multi-component lifted MRD codes}. Note, that our construction differs from the one in \cite{sk08u}. Moreover, the construction of \cite{ga11} is more general and does not give an explicit formula for the size of such codes.


\begin{thm}\label{thm:multi}
Let $C_j$ be some MRD code with minimum rank distance $d$ in $\F_q^{k\times n-k-jd}$ for $j=0,\dots, \lfloor \frac{n-k}{d} \rfloor$. Then   
\[\mathcal{C}_j=\{\textnormal{rs}\left[\begin{array}{lll}0_{k\times j\cdot d} & I_{k\times k} & M \end{array}\right] \mid M \in C_j\}\]
are called the component codes and the union
\[\mathcal{C}=\bigcup_{j=0}^{\lfloor \frac{n-k}{d}\rfloor} \mathcal{C}_j \]  
is an $[n,N ,d,k]_q$-code, where 
\[N=\sum_{i=0}^{\lfloor\frac{n-2k}{d}\rfloor} q^{(k-d+1)(n-k-d i)} + \sum_{i=\lfloor\frac{n-2k}{d}\rfloor+1}^{\lfloor\frac{n-k}{d}\rfloor} \lceil q^{k(n-k+1-d (i+1))} \rceil .\]
If $k=d$ and $n\equiv r \mod k$ (such that $0\leq r < k$), it holds that
\[N        
= \frac{q^n- q^{r+k}}{q^k-1}+1 = q^r(\frac{q^{n-r}-1}{q^k-1}+q^{-r} - 1).\]
\end{thm}

\begin{proof}
We will first prove the minimum distance. It is clear that the distance between any elements of the same component $\mathcal{C}_i$ is greater than or equal to $d$. Now let $U\in \mathcal{C}_i$ and $V\in \mathcal{C}_{i+1}$. Since the identity blocks are shifted by $d$ positions, the maximal intersection is $(k-d)$-dimensional. Thus
\[d_I(U,V)= k-\dim (U\cap V) \geq k - (k-d) = d .\]
Let us now investigate the size of the code. The subspace component code is as large as the corresponding MRD code, thus
\[|\mathcal{C}_i|= \left\{\begin{array}{ll} \lceil q^{(n-k-d i)(k-d+1)}\rceil & \textnormal{ for } n-d i\geq 2k \\  \lceil q^{k(n-k+1-d (i+1))}\rceil & \textnormal{ for } n-d i < 2k \end{array} \right.\]
(follows from Lemma \ref{thm1}).
As $n-d i\geq 2k \Leftrightarrow i\leq \frac{n-2k}{d}$ and we look at codes with $n\geq 2k$, we proved the general formula. 
For $d=k$ it holds that
\[\sum_{i=0}^{\lfloor\frac{n}{k}\rfloor-2} q^{n-k(i+1)} + \sum_{i=\lfloor\frac{n}{k}\rfloor-1}^{\lfloor\frac{n}{k}\rfloor-1} \lceil q^{k(n-k+1-k (i+1))} \rceil = \sum_{i=0}^{\lfloor\frac{n}{k}\rfloor-2} q^{n-k(i+1)} + \lceil q^{k(n-k+1-k \lfloor\frac{n}{k}\rfloor)} \rceil \]
\[ = \frac{q^n-q^{k+n-k\lfloor \frac{n}{k}\rfloor}}{q^k-1}+ \lceil q^{k(n+1-k \lfloor\frac{n}{k}\rfloor)} \rceil .\]
Note, that $n-k\lfloor \frac{n}{k}\rfloor = r$, if $n\equiv r \mod k$. Thus, the expononent of the second summand is non-positive and the formula for $N$ follows.
\end{proof}

We can now obtain a general lower bound on the size of constant dimension codes:

\begin{cor}
 \[A_q(n,d,k) \geq \sum_{i=0}^{\lfloor\frac{n-2k}{d}\rfloor} q^{(k-d+1)(n-k-d i)} + \sum_{i=\lfloor\frac{n-2k}{d}\rfloor+1}^{\lfloor\frac{n-k}{d}\rfloor} \lceil q^{k(n-k+1-d (i+1))} \rceil .\]
\end{cor}

\section{Conclusion}

Lifte rank metric codes are a prominent family of constant dimension codes. Taking unions of such codes of different length (with zero columns appended in front) achieves codes of the same minimum distance but larger cardinality. In this work we give an explicit formula for the size of such codes which gives rise to a lower bound on the size of constant dimension codes $\subseteq \G$ for any $q,k,n$ and minimum distance $d$.

\bibliography{/home/b/trautman/Dropbox/my_bib/network_coding_stuff}
\bibliographystyle{amsplain}

\end{document}